\documentclass[twocolumn]{article}
\usepackage[paper=a4paper,margin=2cm,includehead=false]{geometry}
\usepackage[hyperfootnotes=false]{hyperref}

\usepackage{amsmath,amssymb,amsfonts, amsthm}
\usepackage{algorithmic}
\usepackage{graphicx}
\usepackage{textcomp}
\usepackage{xcolor}
\usepackage{pgfplots}
\usepackage{float}
\usepackage{url}
\usepackage{tikz}
\usepackage{hyperref}
\usepackage{subcaption}
\usetikzlibrary{patterns,backgrounds}
\AtBeginDocument{
	\providecommand{\tikzpicturedependsonfile}[1]{}
	\providecommand{\tikzsetnextfilename}[1]{}
}

\newenvironment{experimentfigure}{\begin{figure*}[tb]}{\end{figure*}}

\pgfplotsset{compat=newest}
\pgfplotsset{minor grid style={dashed,very thin, color=blue!15}}
\pgfplotsset{major grid style={very thin, color=black!30}}
\pgfplotsset{
	every axis title/.style={font=\scriptsize, at={(0.5,1)},anchor=base, yshift=1pt},
	xtick pos=lower,
}
\pgfplotsset{
	automatically generated axis/.style={
		height=105pt,%
		width=174pt,
		scaled ticks=false,
		xticklabel style={font=\tiny,/pgf/number format/.cd, fixed,/tikz/.cd},%
		yticklabel style={font=\tiny,/pgf/number format/.cd, fixed,/tikz/.cd},%
		x label style={at={(ticklabel cs:0.5, -5pt)},name={x label},anchor=north,font=\scriptsize},
		y label style={at={(ticklabel cs:0.5, -5pt)},name={y label},anchor=south,font=\scriptsize},
	},
	automatically generated symbolic/.style={
		height=105pt,
		width=500pt,
		xticklabel style={font=\tiny,rotate=90},
		yticklabel style={font=\tiny,/pgf/number format/.cd, fixed,/tikz/.cd},%
		x label style={at={(ticklabel cs:0.5, -5pt)},name={x label},anchor=north,font=\scriptsize},
		y label style={at={(ticklabel cs:0.5, -5pt)},name={y label},anchor=south,font=\scriptsize},
	},
	first kind/.style={
		legend style={font=\scriptsize,fill=none},
		legend columns=2,legend cell align=left,
	},
	posterior kind/.style={
		legend style={draw=none},
	},
	nines axis/.style={
		height=105pt,
		width=255pt,
		scaled ticks=false,
		xticklabel style={font=\tiny,/pgf/number format/.cd, fixed,/tikz/.cd},
		yticklabel style={font=\tiny,/pgf/number format/.cd, fixed,/tikz/.cd},
		x label style={at={(ticklabel cs:0.5, -5pt)},name={x label},anchor=north,font=\scriptsize},
		y label style={at={(ticklabel cs:0.5, -5pt)},name={y label},anchor=south,font=\scriptsize},
		legend style={font=\scriptsize,fill=white,row sep=-1pt,inner sep=1pt},
		legend cell align=left,
		enlarge x limits={abs=5pt},
		grid=both,
	}
}
\def\timetickcode{%
	\pgfkeys{/pgf/fpu,/pgf/fpu/output format=fixed}%
	\pgfmathparse{\tick}%
	\edef\tmp{\pgfmathresult}%
	\pgfmathtruncatemacro\seconds{\tmp-60*floor(\tmp/60)}%
	\pgfmathtruncatemacro\tmp{(\tmp - \seconds)/60}%
	\pgfmathtruncatemacro\minutes{\tmp-60*floor(\tmp/60)}%
	\pgfmathtruncatemacro\tmp{(\tmp - \minutes)/60}%
	\pgfmathtruncatemacro\hours{\tmp-24*floor(\tmp/24)}%
	\pgfmathtruncatemacro\days{(\tmp - \hours)/24}%
	\ifnum\days=0%
		{\tiny\hours:\minutes:\seconds}%
	\else%
		{\tiny\days-\hours:\minutes}%
	\fi%
}
\def\memorytickcode{%
	\pgfkeys{/pgf/fpu,/pgf/fpu/output format=fixed}%
	\pgfmathtruncatemacro\unitcase{log2(\tick+0.001)/10+1.1}%
	\pgfmathparse{\tick / pow(1024,\unitcase-1)}%
	\pgfmathprintnumber{\pgfmathresult}%
	\ifcase\unitcase B%
		\or KB%
		\or MB%
		\or GB%
		\or TB%
		\else +1024,\pgfmathprintnumber{unitcase}B%
	\fi%
}
\tikzset{
	automatically generated plot/.style={
		/pgfplots/error bars/x dir=both,
		/pgfplots/error bars/y dir=both,
		/pgfplots/error bars/x explicit,
		/pgfplots/error bars/y explicit,
		/pgfplots/error bars/error bar style={ultra thin,solid},
		/tikz/mark options={solid},
	},
	automatically generated bar plot/.style={
		/pgfplots/error bars/y dir=none,
		/pgfplots/error bars/y explicit,
	},
	automatically generated boxplot/.style={
	},
	x time ticks/.style={
		/pgfplots/scaled x ticks=false,
		/pgfplots/xticklabel={\timetickcode},
	},
	y time ticks/.style={
		/pgfplots/scaled y ticks=false,
		/pgfplots/yticklabel={\timetickcode}
	},
	x memory ticks from kilobytes/.style={
		/pgfplots/scaled x ticks=false,
		/pgfplots/xticklabel={\memorytickcode}
	},
	y memory ticks from kilobytes/.style={
		/pgfplots/scaled y ticks=false,
		/pgfplots/yticklabel={\memorytickcode}
	},
}

\tikzset{
	bar group position/.style 2 args={
		/pgf/bar shift={%
			-0.5*(#2*\pgfplotbarwidth + (#2-1)*\pgfkeysvalueof{/tikz/bar group skip}) +
			(.5+#1)*\pgfplotbarwidth + #1*\pgfkeysvalueof{/tikz/bar group skip}
		},
	},
	bar group skip/.initial=2pt,
}

\colorlet{OFT}{green}
\colorlet{OFT MV}{black}
\colorlet{MRLS 1}{red}
\colorlet{MRLS KSP}{blue}
\colorlet{FT 2}{cyan}
\colorlet{MRLS 2}{magenta}
\colorlet{FT 3}{cyan}
\colorlet{MRLS 3}{violet}

\tikzset{
	generated bar/.style={/pgfplots/area legend},
	OFT bar/.style={generated bar,fill=OFT!20,postaction={pattern=grid}},
	OFT MV bar/.style={generated bar,fill=OFT MV!20,postaction={pattern=horizontal lines}},
	MRLS 1 bar/.style={generated bar,fill=MRLS 1!20,postaction={pattern=north west lines}},
	MRLS KSP bar/.style={generated bar,fill=MRLS KSP!20,postaction={pattern=north east lines}},
	FT 2 bar/.style={generated bar,fill=FT 2!20,postaction={pattern=crosshatch}},
	MRLS 2 bar/.style={generated bar,fill=MRLS 2!20,postaction={pattern=dots}},
	FT 3 bar/.style={generated bar,fill=FT 3!20,postaction={pattern=bricks}},
	MRLS 3 bar/.style={generated bar,fill=MRLS 3!20,postaction={pattern=checkerboard}},
	DF+ bar/.style={OFT bar},
	DF bar/.style={FT 2 bar},
	MRLS 1.46 bar/.style={MRLS 1 bar},
	,
	OFT line/.style={OFT,dashed,mark=square,thick,mark options={scale=0.8,solid}},
	MRLS 1 line/.style={MRLS 1,solid,mark=o,thick,mark options={scale=0.8,solid}},
    FT 2 line/.style={FT 2,dashed,mark=diamond,thick,mark options={scale=0.8,solid}},
    MRLS 2 line/.style={MRLS 2,solid,mark=triangle,thick,mark options={scale=0.8,solid}},
	DF+ line/.style={OFT line},
	DF line/.style={FT 2 line},
	MRLS 1.46 line/.style={MRLS 1 line},
	MRLS 3 line/.style={MRLS 3,solid,mark=square,thick,mark options={scale=0.8,solid}},
	FT 3 line/.style={FT 3,dashed,mark=diamond,thick,mark options={scale=0.8,solid}},
}

\newtheorem{theorem}{Theorem}[section]
\newtheorem{definition}[theorem]{Definition}

\colorlet{region 2}{red!20}
\colorlet{region 3}{blue!20}
\colorlet{region 4}{green!20}
\colorlet{region 5}{violet!20}
\colorlet{region 6}{yellow!20}

\def\costlinks{\ensuremath{\mathrm{Cost}_{\mathrm{links}}}}
\def\costswitches{\ensuremath{\mathrm{Cost}_{\mathrm{switches}}}}

\DeclareMathOperator{\rem}{rem}

\begin{document}

	\title{Extreme-Scale Interconnection Networks}

	\author{Alejandro Cano$^\dagger$, Cristina Brinza$^\ddagger$, Cristóbal Camarero$^\dagger$, Carmen Martínez$^\dagger$, Ramón Beivide$^{\dagger,\ast}$\\
		$^\dagger$\textit{Universidad de Cantabria}, SPAIN\\
		$^\ddagger$Performed while working for the \textit{Universidad de Cantabria}, SPAIN\\
		$^\ast$\textit{Barcelona Supercomputing Center}, SPAIN\\
		\small\{alejandro.cano, cristobal.camarero, carmen.martinez,ramon.beivide\} @unican.es, brinza.cristina@outlook.com
	}

\maketitle

	\begin{abstract}
		Extreme-scale data centers are the backbone of next-generation computing, enabling breakthroughs in science, artificial intelligence, and global innovation through unprecedented processing power and scalability. This work examines leaf-spine network topologies that offer extreme scalability—connecting a vast number of endpoints—while delivering strong performance at low cost. It takes as a starting point two alternatives to the widely used Fat-Tree topology: the Orthogonal Fat-Tree and the Random Folded Clos. The resulting Multipass Random Leaf-Spine (MRLS) networks inherit their advantages and surpass Fat-Trees in both throughput and flexibility. To fully leverage the topological properties of these networks, various non-minimal routing strategies are considered. An exhaustive evaluation using an interconnection network simulator provides insight into the trade-offs and scalability of these topologies under realistic conditions, positioning them as a promising solution for extreme-scale systems.
		 The MRLS achieves a 50\% speedup against a Fat-Tree for an All2All collective comprising 100k endpoints, and 100\% against Dragonfly networks for the same collective.
	\end{abstract}

	\section{Introduction}

	As workload demands in Data Center (DC) and High-Performance Computing (HPC) systems continue to increase, the scale of network interconnects also expands. This growth directly impacts the achievable performance of the overall system. Extreme-scale DCs interconnect hundreds of thousands of endpoints, with projections of their upgrade to millions~\cite{citaextremescale,SemiAnalysis,Pilz}.

	A significant proportion of the interconnects used nowadays are indirect networks such as the Fat-Tree (FT) or Folded Clos~\cite{top500}.
In an indirect network (IN), switches are organized into levels.
The lowest level consists of \textsl{leaf switches}, which connect directly to endpoints and to switches in the level above.
The highest level consists of \textsl{root switches} (or \textsl{spines} if it's a two-level network), which connect only to the switches below them.
Communication between endpoints is \textsl{indirect}, meaning traffic must traverse intermediate switches that are not connected to any endpoint.
In contrast, in direct networks (DN) such as Dragonfly (DF)~\cite{Kim}, there is no distinction among switches, and all of them have endpoints connected.

	The FT is a particularly notable indirect network for being rearrangeably non-blocking, meaning it can support communication between any pair of nodes using disjoint paths.
	This implies there is no interference between multiple transmissions when employing circuit switching. %
	However, the implementation cost of FT networks is substantially higher compared to other alternatives, prompting the exploration of more scalable and cost-effective solutions.
	For extreme-scale systems, where the network can account for a significant fraction of the system cost, minimizing it is crucial, even if it poses additional challenges.

	Among scalable 2-level network topologies, the Orthogonal Fat-Tree (OFT) is a prominent candidate~\cite{oft}. Figure~\ref{fig:OFT} shows the smallest one. While it offers much greater scalability than FT, its performance degrades significantly when subjected to non-uniform traffic. To address this, non-minimal routing strategies are often employed to improve throughput and load balancing~\cite{d2}.
	However, a key limitation of the OFT lies in its rigid construction requirements: it can only be built for very specific parameter values. This constraint limits both the applicability and expandability of the topology, making it difficult to adapt to the desired system size.

	\begin{figure}[tb]
	\centering%
	\begin{tikzpicture}[y=0.9cm,x=0.6cm,font=\scriptsize,
	]
	\foreach \a in {0,1}
	\foreach \b in {0,1}
	\foreach \c in {0,1}
	{
		\pgfmathtruncatemacro\heigh{\a*4+\b*2+\c}
		\ifthenelse{\heigh=0}{}
		{
			\path[fill] (\heigh,0) node[anchor=south] {} circle (2pt) coordinate (left \a \b \c)
				+(7,0) coordinate (right \a \b \c) circle (2pt) node [anchor=south] {}
				+(\heigh-.5,1) coordinate (above \a \b \c) circle (2pt) node [anchor=south] {}
				;
			\foreach \s in {-1,0,1}
			{
				\draw (left \a \b \c) -| +(\s*4pt,-7pt) circle (1.5pt);
				\draw (right \a \b \c) -| +(\s*4pt,-7pt) circle (1.5pt);
			}
		}
	}
	\foreach \la in {0,1}
	\foreach \lb in {0,1}
	\foreach \lc in {0,1}
	\foreach \ra in {0,1}
	\foreach \rb in {0,1}
	\foreach \rc in {0,1}
	{
		\pgfmathtruncatemacro\goodl{\la==1 || \lb==1 || \lc==1}
		\pgfmathtruncatemacro\goodr{\ra==1 || \rb==1 || \rc==1}
		\pgfmathtruncatemacro\dot{mod(\la*\ra + \lb*\rb +\lc*\rc,2)}
		\pgfmathtruncatemacro\good{\goodl && \goodr && \dot==0}
		\ifthenelse{\good=1}
		{
			\draw (left \la \lb \lc) -- (above \ra \rb \rc) -- (right \la \lb \lc);
		}
	}
	\end{tikzpicture}%
	\caption{Orthogonal Fat Tree with parameter $q=2$. It has $n=14$ leaf switches and $k=3$ up-links at the leaf switches.}
	\label{fig:OFT}
\end{figure}
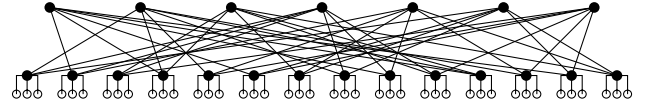

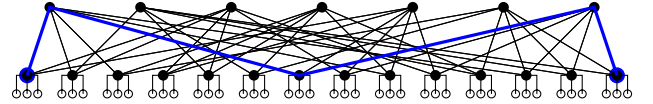
\begin{figure}[tb]
	\centering%
	\begin{tikzpicture}[y=0.9cm,x=0.6cm,font=\scriptsize]
	\path[fill=blue] (0,0) circle (3pt) (13,0) circle (3pt);%
	\foreach \x in {0,...,13}
	{
		\path[fill] (\x,0) circle (2pt) coordinate(nodo \x);
		\foreach \s in {-1,0,1}
		{
			\draw (nodo \x) -| +(\s*4pt,-7pt) circle (1.5pt);
		}
	}
	\foreach \x in {14,...,20}
	{
		\pgfmathparse{2*(\x-14)+.5}
		\path[fill] (\pgfmathresult,1) circle (2pt) coordinate(nodo \x);
	}
	\foreach \x/\A in {
		0/{16,18,14},
		1/{14,15,16},
		2/{14,17,19},
		3/{17,18,16},
		4/{15,16,17},
		5/{19,17,15},
		6/{20,18,14},
		7/{20,16,18},
		8/{18,17,14},
		9/{16,18,20},
		10/{19,17,15},
		11/{20,14,19},
		12/{20,19,15},
		13/{20,19,15},
		14/{0,1,2,6,8,11},
		15/{1,4,5,10,12,13},
		16/{0,1,3,4,7,9},
		17/{2,3,4,5,8,10},
		18/{0,3,6,7,8,9},
		19/{2,5,10,11,12,13},
		20/{6,7,9,11,12,13}}
	{
		\foreach \y in \A
		{
			\draw (nodo \x) -- (nodo \y);
		}
	}
	\draw[very thick,blue] (nodo 0) -- (nodo 14) -- (nodo 6) -- (nodo 20) -- (nodo 13);
	\end{tikzpicture}%
	\caption{MRLS with $n=14$ leaf switches and $k=3$ up-links at the leaf switches. Highlighted a pair of nodes at distance 4 and one path.}\label{fig:RFC}
\end{figure}

The Random Folded Clos (RFC)~\cite{rfc} offers an alternative topology designed to balance performance and scalability while facilitating flexible system expansion.
Unlike the strictly defined interconnection pattern of the OFT, the RFC employs a random interconnection pattern between switch levels.

Figure~\ref{fig:RFC} depicts a 2-level, randomly interconnected network featuring the same number of switches and endpoints as the OFT shown in Figure~\ref{fig:OFT}.
However, this topology deviates from the standard RFC, as it lacks inherent up-down connectivity and allows minimal path lengths of up to four hops.
Consequently, we introduce this distinct topology which we call \textsl{Multipass Random Leaf-Spine} (MRLS), alongside a corresponding routing mechanism.
This work investigates such networks and evaluates their suitability for extreme-scale systems.

In particular, our contributions are as follows:
\begin{enumerate}
    \item The proposal of MRLS, a novel network topology, alongside a routing method for randomly wired indirect networks.
    \item A comprehensive analysis of the architectural benefits of MRLS, demonstrating key topological advantages such as fine-grained scalability and enhanced routing efficiency.
    \item An extensive performance evaluation comparing MRLS networks against state-of-the-art indirect and direct topologies such as Dragonfly, including large-scale simulations supporting up to 100,000 endpoints.
\end{enumerate}

The structure of the paper is as follows. Section~\ref{sec:background} introduces the notation for topologies and formulae for performance and cost. Section~\ref{sec:motivation} elaborates on the limitations of current topologies. In Section~\ref{sec:MRLS}, the MRLS proposal is formally introduced and studied. Section~\ref{sec:setup} describes the workloads for simulation and selected scenarios. The simulation results are shown and analyzed in Section~\ref{sec:results}. Section~\ref{sec:conclusions} concludes the paper.

\section{Background}\label{sec:background}

	This section establishes the basic concepts and notation used in the paper.
We begin by defining the network models and standard topologies, followed by a formal analysis of the relationships between topological distances, performance, and cost.
Table~\ref{tbl:notation} summarizes the key notation employed throughout the article.

\subsection{Topologies}

\begin{table}
	\centering%
	\caption{Notation throughout the article.}%
	\label{tbl:notation}%
	\begin{tabular}{|l|p{.8\linewidth}|}
	\hline
	$S$	&	Total number of endpoints.\\
	$N$	&	Total number of switches.\\
	$M$	&	Total number of bidirectional links.\\
	$R$	&	Switch radix.\\
	$A$	&	Average distance between leaf switches in an IN, or between vertices in a DN.\\
	$D$	&	Diameter, the maximum distance between leaf switches.\\
	$d(x,y)$ & Distance from $x$ to $y$.\\
	$\Theta$ & Estimator of the load per endpoint.\\
	$l$	&	Number of switch levels.\\
	$h$	&	Height of a multistage network, $h=l-1$.\\
	$N_i$	&	Number of switches in the level $i$.\\
	$u$	&	Number of upward ports of a switch in an IN.\\
	$d$	&	Number of downward ports of a switch in an IN.\\
	$f$	&	Thickness factor of the switches in an IN.\\
	$q$	&	A power of prime, employed in OFT (cf. \cite{projective,slimfly}).\\
	\hline
	\end{tabular}
\end{table}

The \textsl{topology} of a network describes the interconnection pattern between endpoints and switches.
In many high-performance systems, each endpoint possesses a single network interface and is connected to one switch.
This allows us to omit the endpoints and abstract the topology as a graph where vertices represent switches and edges represent the physical links between them.

Although this abstraction simplifies the model, it is robust.
Specialized architectures like BCube~\cite{BCube} (multi-interface endpoints) or GPU-dense nodes like those in the Frontier supercomputer (multiple NICs per endpoint) technically deviate from this simple model.
However, these cases represent only a constant-factor increase in connectivity.
Therefore, we proceed with the switch-level graph model, noting that extensions to multi-interface scenarios are straightforward.

    Next, FTs and other more cost-effective topologies are introduced.

\subsubsection{The non-blocking Fat-Tree}
The most prevalent topology in modern DCs and HPC systems is the Fat-Tree~\cite{Leiserson}, which is essentially a multistage folded Clos network~\cite{Clos}.

To ensure non-blocking behavior, standard Fat-Trees are built such that the aggregate bandwidth going upwards from any switch equals the bandwidth going downwards.
Using switches with radix $R$, this implies splitting ports evenly, using $R/2$ ports for \textit{up} links and the other $R/2$ for \textit{down} links.

For a Fat-Tree of height $h$:
\begin{itemize}
	\item There are $h+1$ levels of switches.
	\item The network supports $S=2(R/2)^{h+1}$ endpoints.
	\item The diameter is $2h$, and the average distance $A$ is close to the diameter.
\end{itemize}

While exhibiting high performance, the Fat-Tree is expensive due to the high number of cables and switches required to maintain its non-blocking bandwidth.

\subsubsection{Cost-Effective topologies}
Driven by the high cost of Fat-Trees, extreme-scale systems often employ other topologies to minimize it when interconnecting the same number of endpoints.
Examples of topologies that have been proposed to reduce the cost of Fat-Trees include:

\begin{itemize}
	\item Slimmed Fat-Tree \cite{Kamil,Leon,Jokanovic,Navaridas}.
	\item Orthogonal Fat-Tree (OFT)~\cite{oft}.
	\item Random Folded Clos (RFC)~\cite{rfc}.
	\item Dragonfly~\cite{Kim}.
\end{itemize}

All these topologies relax the non-blocking assumption of the FT.
The OFT, RFC, and Dragonfly reduce the average distance with respect to FTs, allowing for a decrease in cost at the expense of performance degradation under non-uniform traffic patterns.

\subsection{Basic Metrics}
The design of any interconnection network is governed by a tight interplay between physical constraints together with performance and cost metrics. Here, we formalize these relationships.

\subsubsection{Topological Distances}
We consider two primary distance metrics measured over the switch-to-switch graph:
\begin{itemize}
\item The \textsl{diameter ($D$)}, the maximum shortest path between any pair of switches connected to endpoints.
\item The \textsl{average distance ($A$)}, the average minimal path length across all pairs of switches connected to endpoints.
\end{itemize}
Each of these metrics has a counterpart considering paths between any pair of switches instead of just between switches connected to endpoints.
These are denoted in the paper as $D^*$ and $A^*$, respectively, and acquire relevance when dealing with evolved packet routing mechanisms such as derouting through an arbitrary switch.

If the routing algorithm allows non-minimal paths, it is important to note that the length of the employed paths may be different in practice.
In many cases, there is a \textsl{maximum path length}, which is a bound on the number of switches that can be visited along a route.
This is a routing-dependent metric complementary to the diameter.

\subsubsection{Performance}

In practice, the performance of an interconnection network is determined by the \textsl{injected load ($L$)} per endpoint\footnote{This is equivalent to the accepted load per endpoint, or throughput.},
which represents the average amount of traffic an endpoint injects into the network normalized by the link bandwidth.
A value of $L=1$ means that all endpoints inject traffic into the network at full rate without provoking saturation.
Under our assumptions, by construction, it is impossible for an endpoint to inject traffic over its nominal rate, that is, with $L > 1$.

While $L$ is a direct measure of performance, it depends on many factors, including the application's communication patterns and the routing algorithm, making it difficult to model directly.
Consequently, it is useful to set an upper bound on $L$ defined only by the network's topological properties in a specific traffic scenario.
We denote it as the network \textsl{capacity limit ($\Theta$)}, which is the performance bound of the topology under uniform random traffic.
Observe that other traffic patterns could surpass $\Theta$.
As established in prior work~\cite{projective,Singla14}, this capacity limit is given by:

\begin{equation}\label{eq:general}
	\Theta=\frac{2M}{S \cdot A}
\end{equation}
where $S$ is the number of endpoints, $M$ is the number of physical links (with $2M$ representing the total full-duplex bandwidth), and $A$ is the average distance.
That is, assuming a constant target capacity $\Theta$, increasing the number of endpoints $S$ without drastically increasing the number of links $M$ requires minimizing the average distance $A$.
This motivation drives the interest in topologies that approach the \textsl{Moore Bound}~\cite{Hoffman}, which defines the theoretical minimum diameter for a graph of a given degree.

The capacity limit serves to classify networks in the following way:

\begin{itemize}
\item A value of $\Theta = 1$ indicates a \textsl{balanced} network, which has enough capacity to support all endpoints transmitting at their full rate. %
\item A value of $\Theta < 1$ signifies an \textsl{oversubscribed} network, where the topology itself is the bottleneck. %
\item A value of $\Theta > 1$ could suggest an \textsl{overdimensioned} network, having more capacity than needed to support full-rate traffic from all endpoints. Nevertheless, this exceeded capacity can be beneficial for handling non-uniform traffic patterns and providing robustness against congestion. The augmented capacity can be leveraged by adaptive routing algorithms, which may use non-minimal paths—increasing the effective path length to a value higher than $A$—to balance the load across the network.
\end{itemize}

Under these conditions, the relationship $0 \leq L \leq \min\{1, \Theta\}$ holds.
Here, achieving $L = \min\{1, \Theta\}$ is also challenging.
Even in non-oversubscribed networks ($\Theta \geq 1$), achieving the maximum injected load of $L=1$ requires additional resources.
Factors like Head-of-Line (HoL) blocking or poor load balancing due to packet routing can prevent the system from reaching its full potential.
Ideal routing and Virtual Output Queueing~\cite{principles} would allow the system to reach $L = \min\{1, \Theta\}$; however, this is not always realistic.
Thus, while having $\Theta \geq 1$ is a necessary condition for achieving good performance, it is not a guarantee.

\subsubsection{Scalability}

The \textsl{scalability} of a topology is its ability to increase the number of endpoints $S$ while maintaining a capacity limit $\Theta$.

Usually, to increase the number of endpoints in a topology, either the number of switches or their radix $R$ and the number of links connecting them, must also be increased.
Also, the capacity limit $\Theta$ is typically set to 1, representing a balanced network. However, other values can be considered based on specific performance requirements.
For example, a fully populated Fat-Tree has a scalability function of  $S=2^{-h}R^{h+1}$, with $\Theta=1.0$.
In contrast, a simple ring topology is considered \textsl{unscalable} because maintaining capacity $\Theta$ requires the endpoint count to freeze as the network grows (since $A$ grows linearly with $N$).

\subsubsection{Cost}
Finally, we quantify the cost of a topology by means of very simple metrics such as the number of links and switches required by each endpoint.
They can be directly obtained from Eq.~(\ref{eq:general}).

The \textsl{Link Cost} per endpoint can be expressed in terms of the average distance:
\begin{equation}\label{eq:cost}
\costlinks = \frac{M}{S} = \frac{\Theta \cdot A}{2}
\end{equation}

Similarly, the \textsl{Switch Cost} per endpoint depends on how many endpoint-facing ports versus inter-switch ports are used.
If we denote $k$ as the average degree of the switch-to-switch graph, the cost due to switches is
\begin{equation}\label{eq:cost_switches}
\costswitches = \frac{N}{S} = \frac{\Theta \cdot A}{k}.
\end{equation}
These equations show that low-diameter topologies, which provide low $A$ values, are essential for building cost-effective, extreme-scale systems.

\section{Motivation}\label{sec:motivation}

There are several reasons why the research conducted in this work is necessary. The most important ones are related to network cost and scalability, which are considered below.

\subsection{Network Cost}

Extreme-scale DCs demand interconnection networks that connect vast numbers of endpoints while balancing high performance against manageable costs. A natural question arises: \textit{How can we minimize the cost per endpoint for a specific system size and switch technology?}

The cost of a FT grows rapidly with the number of endpoints, as $\costlinks=h$ and $\costswitches=\tfrac{2}{R}(h+\tfrac12)$, making them impractical for extreme-scale systems.
With switches of 64 ports ($R=64$), a 2-level Fat-Tree can support up  to 2048 endpoints, and a 3-level Fat-Tree can support up to 65,536 endpoints.
To surpass this number, one would need to increase the number of levels to 4, leading to a prohibitive cost in cables (50\% more than in a 3-level FT), switches (40\% more), and energy consumption.

Cost reduction is often achieved through coarse adjustments, such as \textsl{oversubscription}.
For instance, in a FT, reducing the number of non-leaf switches cuts costs but immediately halves the bisection bandwidth.
Hence, one must choose between a generic, expensive non-blocking FT or a significantly cheaper, lower-performance slimmed FT.

Similarly, direct networks like Dragonfly~\cite{Kim} reduce costs by optimizing for uniform traffic, but this structural decision fixes the performance characteristics.
If the workload demands slightly better worst-case performance, the network cannot be easily tuned up, requiring a different topological class entirely.

Therefore, to achieve extreme scalability while maintaining good performance at an attainable cost, other, more cost-efficient topologies must be considered.

\subsection{Rigidity of network topologies}

A major drawback of most cost-efficient topology families is the discrete nature of their valid network sizes, suffering from structural rigidity.
If the system needs to be expanded in the future, it may be very challenging and costly to find a new valid configuration.
Also, performance is usually rigid within a given topology family.
To overcome these sizing constraints, network topologies should be flexible, with fine-grained scalability and performance.

\subsubsection{Fine-grain scalability}
Fine-grain scalability is intrinsically linked to \textsl{incremental expandability}~\cite{Zhao} (also referred to as fine-grained incremental expansion).
This property ensures that a DC can grow its capacity by small increments—adding a single rack or switch—rather than requiring massive, step-wise upgrades to reach the next valid topological size.
Randomly wired topologies indeed satisfy these criteria satisfactorily, as described in \cite{Singla,random18}.

It is important to separate network scalability from the practice of connecting fewer endpoints than possible.
Almost any topology can support any number of endpoints $S$ below its maximum by simply leaving ports unconnected (depopulation).
The following definition formalizes, in our context, the fine-grain scalability concept:

\begin{definition}\label{def:finegrain}
	The \textsl{fine grain scalability} of a topology is the density of the achievable $(N_1,R)$ pairs, with $N_1=N$ for direct networks. This density is the inverse of the gap between the values of two consecutive instances.
\end{definition}

We prefer to use $N_1$ instead of the whole $N$ as it represents more clearly the goal of adding switches with endpoints, rather than just switches for their own sake.

The FT offers moderate granularity.
While it can be depopulated by removing entire pods, maintaining full bisection bandwidth typically restricts the network to specific sizes.
That is, switches at the top level cannot be removed, as it would lead to a general performance loss.

The OFT is much more rigid.
Due to its reliance on specific routing paths to handle adversarial traffic, leaf switches cannot be selectively removed without breaking connectivity guarantees.
Thus, only full configurations are possible.
The values of radix are $R=2(q+1)$, with $q$ a prime power.
By the Prime Number theorem this means that around some target $R^*$, there are about one valid value from each $2\log R^*$.
Its number of leaf switches $N_1=2(q^2+q+1)$ gives an additional restriction.
Around some target $N_1^*$ there will be about one valid value for each $\sqrt{2N_1^*}\log N_1^*$.

\subsubsection{Fine-grain performance}
In standard topologies, not only is the size of the systems is rigid, but also their performance
This work explores extreme-scale topologies that not only are fine-grained scalable but also allow for \textsl{fine-grain performance tuning}.
This implies that for any desired number of endpoints $S$ and switch radix $R$, the topology can be adjusted to achieve a specific performance level, rather than being constrained to a few discrete options.

\begin{figure*}
	\centering%
	\input{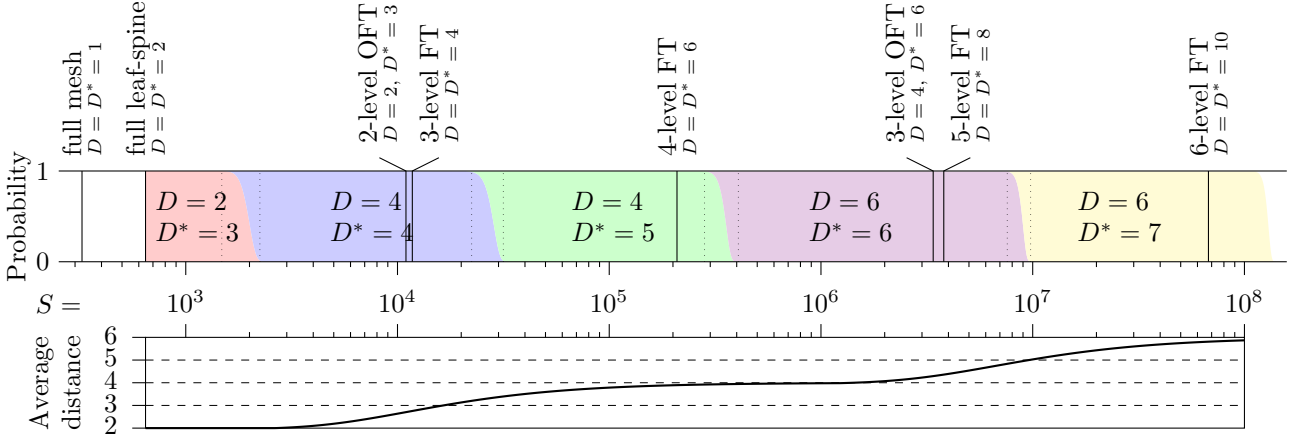}
	\caption{Scalability spectrum for MRLS with $R=36$ and $f=1$ together some references.
	Regions are marked with $D=x$ $D^*=y$, meaning that a MRLS in the region has a leaf to leaf diameter of $x$ and a maximum distance between arbitrary switches of $y$.
	}
	\label{fig:scal-spectrum}
\end{figure*}

\section{Multipass Random Leaf Spine Networks}\label{sec:MRLS}

Among other interesting properties, regular randomly wired direct topologies such as Jellyfish, exhibit lower cost and finer grain scalability than FTs, as described in \cite{Singla} .
In \cite{random18}, the Random Folded Clos network was introduced as an indirect alternative to the Jellyfish.
RFC networks were designed to use Up/Down routing.
This simplifies key aspects such as deadlock avoidance, but introduces two important limitations.
Foremost, it entails a notable restriction on scalability.
This can be waived by adding switch levels, but this leads to unnecessarily large costs.
Secondly, even in those networks that are up/down-connected, there is a performance loss from not considering other routes.

In this work, we propose the \textsl{Multipass Random Leaf Spine} (MRLS) networks.
Our objective is to decouple the physical topology from the routing restrictions, creating a network that retains the cabling simplicity of leaf-spine architectures while unlocking fine-grained scalability and performance.
MRLS networks are conceived for extreme-scale systems.
In these networks, connectedness between certain pairs of leaf switches is guaranteed through multiple up/down routing phases through the network.
Thus, although the network has just 2 levels of switches, its diameter is higher than 2.
We formally introduce these networks in the following definition.

\begin{definition}\label{def:general_LS}
	A \textsl{Multipass Random Leaf-Spine} is a 2-level network which consists of $N_1$ leaf switches and $N_2$ spine switches.
	Each leaf switch has $d$ down links to connect to endpoints and $u$ uplinks connected to varying spine switches.
	Each spine switch has $u+d$ downlinks connected to leaf switches.
	All switches have radix $R=u+d$.
	The interconnections between the two levels of switches are decided via a random process~\cite{Steger}.
\end{definition}

We denote in this work the \emph{thickness} factor as $f=u/d$.
The inverse of this thickness factor is used in other works as the slimming factor~\cite{Navaridas}, which is also known as the \textsl{fitness ratio}, \textsl{contention factor}, or \textsl{blocking ratio}~\cite{slimmedRFC}.
The reasonable thickness factor values for MRLS lie in $1 \leq f\leq 3$, as we show in the next section.

Unlike FTs, RFCs or OFTs, a MRLS network may be built for any value of $N_1$ and $N_2$ following $uN_1=RN_2$.
This is the greatest fine-grain scalability possible for an indirect network.

\subsection{Scalability}

Extreme-scale DCs require topologies capable of expanding from $100{,}000$ to potentially millions of endpoints~\cite{Pilz}.
In this context, we compare MRLS networks against known topologies (FT, OFT, RFC).

In Figure~\ref{fig:scal-spectrum}, we illustrate the expandability spectrum of MRLS networks for a fixed radix $R=36$.
The employed thickness factor is $f=1$.
The figure highlights regions where MRLS networks can achieve specific even diameters $D$ depending on the number of network passes needed to connect a number of endpoints $S$.
The diameter including spines ($D^*$) is also a notable structural property of the topology and is important for the non-minimal routing algorithms required by these topologies.
Hence, their thresholds are included.
Observe that two consecutive regions for $D^*$ are included in a region with a given diameter $D$.
It should be noted that $D$ can be obtained from $D^*$ by $D=2\lfloor\tfrac{D^*}{2}\rfloor$, but not \textit{vice versa}.
The boundaries between regions are actual function plots, indicating the probability for a MRLS instance with the indicated number of endpoints to lie in the left region.
Also, they are quite sharp, with a small proportion of sizes being in reasonable doubt about which side they will fall on.
Technical details can be found in Appendix~\ref{apx:probability}.

Looking at the figure, the first boundary is observed to lie around 2K endpoints with a reasonable chance of having a diameter of either 2 or 4.
Having $D=2$ establishes the limit up to which up/down routing can be employed.
From this limit, multi-pass routing is needed, incrementing the diameter by 2 units for each pass.
The next boundary is around 30K endpoints, in which the diameter does not change, but only $D^*$.
That is, the leaf-to-spine maximum distance changes from 4 to 5.
To the left, the 2-level OFT can be seen, which has parameters $D=2$ and $D^*=3$; and also the 3-level FT, with $D=D^*=4$.
Looking at the far right, it can be observed that even using a small radix such as $R=36$, a MRLS network can connect 100M endpoints with $D=6$, way beyond current practical goals.

The bottom part of Figure~\ref{fig:scal-spectrum} shows the value of the expected average distance $A$ for a MRLS network with a given number of endpoints $S$.
It can be seen to grow slowly, giving $2\leq A\leq 6$ for any practical size.
As $f=1$ is employed, these correspond to values $\tfrac13\leq \Theta\leq 1$ for the capacity limit.

	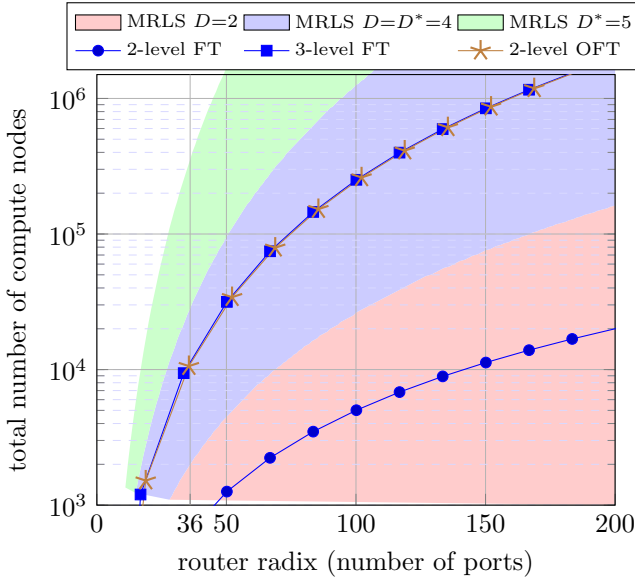
\begin{figure}
		\centering%
				\begin{tikzpicture}
			\begin{semilogyaxis}[
				xmin=0,xmax=200,
				ymin=1e3,ymax=1.5e6,
				domain=0.1:200,
				enlargelimits=false,
				xmajorgrids=true,
				ymajorgrids=true,
				xminorgrids=true,
				yminorgrids=true,
				extra x ticks=36,
				minor y tick num=1,
				minor grid style={dashed,very thin, color=blue!15},
				major grid style={very thin, color=black!30},
				xlabel={router radix (number of ports)},
				ylabel={total number of compute nodes},
				legend style={at={(0.50,1.01)},anchor=south,font=\scriptsize},legend columns=3,legend cell align=left,
				legend image post style={every path={nomorepostactions}},
			]
				\def\thrcoorA{(28, 1096.478196) (29, 1174.897555) (30, 1288.249552) (31, 1380.384265) (32, 1513.561248) (33, 1621.810097) (34, 1737.800829) (35, 1862.087137) (36, 1995.262315) (37, 2137.962090) (38, 2290.867653) (39, 2398.832919) (40, 2570.395783) (41, 2754.228703) (42, 2884.031503) (43, 3090.295433) (44, 3235.936569) (45, 3467.368505) (46, 3630.780548) (47, 3801.893963) (48, 4073.802778) (49, 4265.795188) (50, 4466.835922) (51, 4677.351413) (52, 4897.788194) (53, 5128.613840) (54, 5370.317964) (55, 5623.413252) (56, 5888.436554) (57, 6165.950019) (58, 6456.542290) (59, 6760.829754) (60, 7079.457844) (61, 7413.102413) (62, 7762.471166) (63, 7943.282347) (64, 8317.637711) (65, 8709.635900) (66, 8912.509381) (67, 9332.543008) (68, 9772.372210) (69, 10000.) (70, 10471.28548) (71, 10715.19305) (72, 11220.18454) (73, 11481.53621) (74, 12022.64435) (75, 12302.68771) (76, 12882.49552) (77, 13182.56739) (78, 13803.84265) (79, 14125.37545) (80, 14791.08388) (81, 15135.61248) (82, 15488.16619) (83, 16218.10097) (84, 16595.86907) (85, 16982.43652) (86, 17782.79410) (87, 18197.00859) (88, 18620.87137) (89, 19054.60718) (90, 19952.62315) (91, 20417.37945) (92, 20892.96131) (93, 21379.62090) (94, 22387.21139) (95, 22908.67653) (96, 23442.28815) (97, 23988.32919) (98, 24547.08916) (99, 25118.86432) (100, 26302.67992) (101, 26915.34804) (102, 27542.28703) (103, 28183.82931) (104, 28840.31503) (105, 29512.09227) (106, 30199.51720) (107, 30902.95433) (108, 31622.77660) (109, 32359.36569) (110, 33113.11215) (111, 33884.41561) (112, 34673.68505) (113, 35481.33892) (114, 36307.80548) (115, 37153.52291) (116, 38018.93963) (117, 38904.51450) (118, 39810.71706) (119, 40738.02778) (120, 41686.93835) (121, 42657.95188) (122, 43651.58322) (123, 44668.35921) (124, 45708.81896) (125, 46773.51413) (126, 47863.00923) (127, 48977.88194) (129, 50118.72336) (130, 51286.13840) (131, 52480.74602) (132, 53703.17964) (133, 54954.08739) (134, 56234.13252) (136, 57543.99373) (137, 58884.36554) (138, 60255.95861) (139, 61659.50019) (140, 63095.73445) (142, 64565.42290) (143, 66069.34480) (144, 67608.29754) (145, 69183.09709) (147, 70794.57844) (148, 72443.59601) (149, 74131.02413) (150, 75857.75750) (152, 77624.71166) (153, 79432.82347) (154, 81283.05162) (156, 83176.37711) (157, 85113.80382) (159, 87096.35900) (160, 89125.09381) (161, 91201.08394) (163, 93325.43008) (164, 95499.25860) (166, 97723.72210) (167, 100000.) (168, 102329.2992) (170, 104712.8548) (171, 107151.9305) (173, 109647.8196) (174, 112201.8454) (176, 114815.3621) (177, 117489.7555) (179, 120226.4435) (181, 123026.8771) (182, 125892.5412) (184, 128824.9552) (185, 131825.6739) (187, 134896.2883) (188, 138038.4265) (190, 141253.7545) (192, 144543.9771) (193, 147910.8388) (195, 151356.1248) (197, 154881.6619) (198, 158489.3192) (200, 162181.0097)};
				\def\thrcoorB{(129, 3.235936569*10^6) (128, 3.162277660*10^6) (127, 3.090295433*10^6) (126, 2.951209227*10^6) (125, 2.884031503*10^6) (124, 2.818382931*10^6) (123, 2.691534804*10^6) (122, 2.630267992*10^6) (121, 2.570395783*10^6) (120, 2.454708916*10^6) (119, 2.398832919*10^6) (118, 2.344228815*10^6) (117, 2.238721139*10^6) (116, 2.187761624*10^6) (115, 2.089296131*10^6) (114, 2.041737945*10^6) (113, 1.995262315*10^6) (112, 1.905460718*10^6) (111, 1.862087137*10^6) (110, 1.778279410*10^6) (109, 1.737800829*10^6) (108, 1.659586907*10^6) (107, 1.621810097*10^6) (106, 1.548816619*10^6) (105, 1.513561248*10^6) (104, 1.445439771*10^6) (103, 1.380384265*10^6) (102, 1.348962883*10^6) (101, 1.288249552*10^6) (100, 1.258925412*10^6) (99, 1.202264435*10^6) (98, 1.148153621*10^6) (97, 1.122018454*10^6) (96, 1.071519305*10^6) (95, 1.023292992*10^6) (94, 1.000000*10^6) (93, 954992.5860) (92, 912010.8394) (91, 870963.5900) (90, 851138.0382) (89, 812830.5162) (88, 776247.1166) (87, 741310.2413) (86, 707945.7844) (85, 676082.9754) (84, 645654.2290) (83, 616595.0019) (82, 588843.6554) (81, 562341.3252) (80, 537031.7964) (79, 512861.3840) (78, 489778.8194) (77, 467735.1413) (76, 446683.5922) (75, 426579.5188) (74, 407380.2778) (73, 389045.1450) (72, 363078.0548) (71, 346736.8505) (70, 331131.1215) (69, 316227.7660) (68, 295120.9227) (67, 281838.2931) (66, 263026.7992) (65, 251188.6432) (64, 234422.8815) (63, 223872.1139) (62, 208929.6131) (61, 199526.2315) (60, 186208.7137) (59, 173780.0829) (58, 165958.6907) (57, 154881.6619) (56, 144543.9771) (55, 134896.2883) (54, 125892.5412) (53, 117489.7555) (52, 109647.8196) (51, 102329.2992) (50, 95499.25860) (49, 87096.35900) (48, 81283.05162) (47, 75857.75750) (46, 69183.09709) (45, 64565.42290) (44, 58884.36554) (43, 54954.08739) (42, 50118.72336) (41, 45708.81896) (40, 41686.93835) (39, 38018.93963) (38, 34673.68505) (37, 31622.77660) (36, 28183.82931) (35, 25703.95783) (34, 22908.67653) (33, 20417.37945) (32, 18620.87137) (31, 16595.86907) (30, 14454.39771) (29, 12882.49552) (28, 11481.53621) (27, 10000.) (26, 8709.635900) (25, 7585.775750) (24, 6456.542290) (23, 5623.413252) (22, 4786.300923) (21, 4073.802778) (20, 3388.441561) (19, 2818.382931) (18, 2344.228815) (17, 1905.460718) (16, 1548.816619) (15, 1230.268771)};
				\def\thrcoorC{(11,1348.963) (12, 2041.737945) (13, 2951.209227) (14, 4168.693835) (15, 5754.399373) (16, 7762.471166) (17, 10232.92992) (18, 13489.62883) (19, 17378.00829) (20, 22387.21139) (21, 28183.82931) (22, 34673.68505) (23, 42657.95188) (24, 52480.74602) (25, 63095.73445) (26, 75857.75750) (27, 91201.08394) (28, 109647.8196) (29, 128824.9552) (30, 151356.1248) (31, 173780.0829) (32, 204173.7945) (33, 234422.8815) (34, 269153.4804) (35, 309029.5433) (36, 354813.3892) (37, 407380.2778) (38, 457088.1896) (39, 524807.4603) (40, 588843.6554) (41, 660693.4480) (42, 741310.2413) (43, 831763.7711) (44, 912010.8394) (45, 1.023292992*10^6) (46, 1.122018454*10^6) (47, 1.258925412*10^6) (48, 1.380384265*10^6) (49, 1.513561248*10^6) (50, 1.698243652*10^6) (51, 1.862087137*10^6) (52, 2.041737945*10^6) (53, 2.238721139*10^6) (54, 2.398832919*10^6) (55, 2.630267992*10^6) (56, 2.884031503*10^6) (57, 3.162277660*10^6)}
				\begin{scope}[on background layer]
				\fill[region 2] plot coordinates \thrcoorA -- (200,1000) --cycle;
				\fill[region 3] plot coordinates \thrcoorA -- (200,2e6) -- plot coordinates \thrcoorB --cycle;
				\fill[region 4] plot coordinates \thrcoorC -- plot coordinates \thrcoorB --cycle;
				\end{scope}
				\addlegendimage{area legend,fill=region 2}\addlegendentry{MRLS $D{=}2$}
				\addlegendimage{area legend,fill=region 3}\addlegendentry{MRLS $D{=}D^*{=}4$}
				\addlegendimage{area legend,fill=region 4}\addlegendentry{MRLS $D^*{=}5$}
				\addplot+[blue,variable=\k] ({2*k},{2*k^2}); \addlegendentry{2-level FT}
				\addplot+[blue,mark options={blue},variable=\k] ({2*k},{2*k^3}); \addlegendentry{3-level FT}
				\addplot+[brown,mark options={brown,thick},mark=star,mark size=4,variable=\q] ({2*(q+1)},{2*(q+1)*(q^2+q+1)}); \addlegendentry{2-level OFT}
			\end{semilogyaxis}
		\end{tikzpicture}
		\caption{Scalability of low-diameter multistage topologies.}
		\label{fig:scalability}
	\end{figure}

Figure~\ref{fig:scalability} provides another scalability comparison, showing the radix required to connect a given number of endpoints with each topology configuration. Both Figure~\ref{fig:scal-spectrum} and \ref{fig:scalability} have their MRLS networks normalized by $f=1$.
Observe that it is possible to normalize in another way, let us say, with $\Theta=1$ by changing $f$.

\subsection{Cost}

In a MRLS network, the cost in links per endpoint is the thickness factor $f$, as it determines how many links are used to connect endpoints versus how many are used for inter-switch connectivity:
\begin{equation}\label{eq:cost_fitness}
	\costlinks(\mathrm{MRLS}) = f = \frac{u}{d} = \frac{M}{S} = \frac{\Theta \cdot A}{2}
\end{equation}

Although expanding the network by adding more endpoints to the MRLS keeps the cost ratios constant, it increases the average distance $A$ and thus reduces the capacity limit $\Theta$.
So, in order to achieve a fixed performance when expanding the network, as $A$ increases, the thickness factor $f$ can be increased as well.
This means that for a leaf switch the number of uplinks $u$ increases while the number of downlinks $d$ to endpoints is reduced.

    In general, an MRLS network has the following advantages:
\begin{itemize}
	\item The cost is easily tunable. The thickness factor $f$ can be adjusted to meet specific cost/performance targets. In Figure~\ref{fig:scal-spectrum} the cost is 1 link per endpoint as $f=1$ but the network can connect any desired of number of endpoints.
	\item Random wiring provides an almost optimal average distance for any given number of switches and links, minimizing the cost for a target performance~\cite{Bollobas}.
\end{itemize}

As stated above, the thickness factor $f$ can be used to tune the performance of the network even beyond what $A$ requires to be balanced in Equation~\ref{eq:general}.
In typical indirect topologies, Up/Down routing is assumed, leading to minimal routes. In this case, the average path length of packets usually equals the topological average distance $A$.
However, routes in MRLS networks are more similar to those of direct networks, and the average path length of packets can be longer than $A$ depending on the selected routing algorithm.
This raises the question of whether to dimension the network for non-uniform traffic by replacing the value of $A$ in Equation~\ref{eq:general} with the longer effective path lengths of the routes in such a scenario.
Thus, if we set $A$ to the average path length of packets under the most adverse permutation traffic and obtain $M$ for $\Theta=1$, we can build a network that behaves similarly to a non-blocking FT for such a permutation.

	\subsection{Multipass routing in leaf-spine networks}

	Unlike FTs, OFTs, or RFCs, MRLS networks are not designed for Up/Down routing, and thus, they require a more general algorithm that can leverage the path diversity of the topology.
	Since a MRLS network is connected randomly, its routing algorithm must be topology-agnostic.
	However, its bipartite indirect structure allows for a significant simplification of the routing logic, compared to direct networks:
	\begin{itemize}
		\item All routes are symmetrical and follow the same $[Up-Down]^*$ structure.
		\item Deadlock avoidance techniques require half the resources compared to other topology agnostic routing algorithms for direct networks: Ladder or Virtual Ordered Classes~\cite{principles} require 1VC per hop, while MRLS requires 1VC per two hops (Up-Down phase), halving the resources needed.
		\item The maximum hop bound of a practical routing algorithm can be reduced to $2D$.
	\end{itemize}
	Two approaches are considered: $K$-Shortest Paths (KSP) algorithm~\cite{ksp}, and Polarized routing~\cite{polarizado}.
	In both routing mechanisms, a penalty is applied to non-minimal routes so that minimal routes are preferred.

	\subsubsection{K-Shortest Paths (KSP)}
	In KSP routing, multiple shortest paths are precomputed for each source–destination pair.
	Randomization among paths of equal length proves to be important to balance traffic \cite{ksp}, while an adaptive selection mechanism allows packets to choose their final route dynamically after the first hop.

	Experiments show improvements when increasing the value of $K$, the number of shortest paths included in the routing table for each source/destination pair. For the simulation, $K=250$ is employed, though further improvement is still possible. Even tables with $K=8$ (cf.~\cite{Singla}) pose problems to use them in actual switches~\cite{Amazon}. This optimistic value is used for a better comparison with Polarized routing.

	\subsubsection{Polarized routing}

	Polarized routing is an adaptive, non-minimal algorithm conceived as a topology-agnostic solution for direct networks \cite{polarizado}.
	Next, it is adapted to indirect random networks, showing great advantages when using it in particular topologies such as MRLS.

	Polarized routing proceeds hop-by-hop by evaluating candidate links based on a very simple logic.
	Any switch $c$ routing a packet from source $s$ to destination $t$ tracks its distance to both endpoints: $d(c,s)$ and $d(c,t)$.
	Since $s$ and $t$ are leaf switches in a leaf-spine network, both distances have the same parity and vary by exactly $\pm 1$ at every hop.

	This property allows us to map the routing decisions into just four possible link categories to evaluate each possible next link, making it very efficient.
	For any switch $n$ as a next hop candidate from switch $c$, the tuple $(d(n,s) - d(c,s), d(n,t) - d(c,t))$ can only have 4 values:

	\begin{itemize}
		\item \textbf{Forward $(+1, -1)$:} The link moves the packet further from the source and closer to the destination. These are minimal routes and are always allowed.
		\item \textbf{Expansion $(+1, +1)$:} The link moves the packet further from both source and destination. This is a non-minimal deroute, only allowed at the beginning of the path, while $d(c,s) < d(c,t)$.
		\item \textbf{Contraction $(-1, -1)$:} The link moves the packet closer to both endpoints. This deroute is only allowed at the end of the path, once $d(c,s) \geq d(c,t)$.
		\item \textbf{Backtrack $(-1, +1)$:} The link moves the packet closer to the source and further from the destination. These are never allowed.
	\end{itemize}

	To implement this, the exact distance values are not required to be carried in the packet header. Instead, a simple 2-bit label per port is sufficient to classify the links into the previous four categories. The routing table is queried twice (indexed by source and destination) to retrieve a bit that indicates if the neighbor is closer ($1$) or further ($0$). This 2-bit combination directly yields the label, and the algorithm simply filters the allowed labels based on the current location of the packet. Among the permitted links, the routing selects the one with the lowest occupancy, penalizing deroutes (Expansion and Contraction) to prioritize minimal paths.

	The simplicity of Polarized rules does not, a priori, guarantee a valid neighbor of $c$ to continue a route for every $c$, $s$, and $t$.
	The switch $c$ is denoted as a \textsl{corner} relative to $s$ and $t$ if there is no neighbor of $c$ satisfying the changes in the distance tuples allowed by the algorithm.
	To employ it successfully, the absence of such corners is required.
	In the case of random topologies, such as MRLS, the probability of having a corner is negligible in practical cases.
	Nevertheless, all pairs $(s,c)$ can be checked, and the MRLS can be re-rolled if any corner is found.

	Polarized routing also benefits from the leaf-spine structure by selecting routes with low maximum length, the same as in Valiant routing~\cite{valiant} and other non-minimal mechanisms.
	It is proved in Theorem~\ref{theo:polarizado} that the maximum length of the routes is upper bounded by twice the diameter $2D$.%

	\begin{theorem}\label{theo:polarizado}
		The routes provided by Polarized routing in a leaf-spine network have a length $k$ satisfying
		\begin{equation*}
			k \leq 2D^* -2 \leq 2D.
		\end{equation*}
	\end{theorem}
	\begin{proof}
	Consider a Polarized route from source $s$ to destination $t$, both leaf switches.
	In a network with levels, there are no links that lead to another switch at the same distance. This entails that for a current switch $c$, any Polarized hop will move away from $s$ when $d(s,c)<d(c,t)$ and toward $t$ otherwise.

	Denote by $m$ the first switch in the route with $d(s,m)\geq d(m,t)$, that is, the \textsl{middle} switch.
	It is clear that $d(s,m)-d(m,t)\leq 1$, as it cannot change by more than 2 at each hop. Furthermore, as $s$ and $t$ are both leaf switches, $d(s,m)\equiv d(m,t)\pmod 2$, leading to $d(s,m)=d(m,t)$. Call $l=d(s,m)$.
	Let $p$ be the \textsl{previous} switch to $m$ in the route, it has $d(s,p)=l\pm 1$ and $d(p,t)=l\pm 1$.
	If it were $d(s,p)=l+1$ then $d(s,p)=l+1\geq d(p,t)=l\pm 1$, contradicting $m$ being the first such switch.
	The same contradiction happens if $d(p,t)=l-1$.

	Thus, $d(s,p)=l-1$ and $d(p,t)=l+1$, implying $l+1\leq D^*$.
	By the preceding argument, the route is minimal from $s$ to $p$ and from $m$ to $t$. This yields a length of $k=(l-1)+1+l=2l\leq 2(D^*-1)$.
	\end{proof}

\begin{table*}
	\centering
	\small
	\caption{Parameters of the evaluated topologies.}
	\label{tabla:topologias}
	\begin{tabular}{|l|l|l|l|l|r|c|p{2cm}|l|}
		\hline
		\textbf{Topology(R,S)}	& Notes & \textbf{\costlinks}	& \textbf{\costswitches} 	& \textbf{Routing}	& \textbf{Endpoints} 	& \textbf{Diameter} & \textbf{Max Hops (99\% / bound)} & \raisebox{-0.3em}{$\Theta$} 	\\ \hline
		MRLS(36, 11052)	& $u=18$		& 1 					& 0.083						& Pol and KSP 		& 11052 			& 4 				& 4 / 6 & 	0.748\\ \hline
		MRLS(36, 11052)	& $u=21$		& 1.4 					& 0.106						& Pol 		& 11052 			& 4 				& 4 / 6 & 1.029\\ \hline
		MRLS(36, 11664)	& $u=24$		& 2 					& 0.139						& Pol 				& 11664 			& 4					& 4 / 6 	& 1.420\\ \hline
		MRLS(36, 104976)	& $u=18$ 	& 1						& 0.083 					& Pol 				& 104976			& 	4	& 6 / 8	& 0.527 \\ \hline
		MRLS(36, 104976)	& $u=24$ 	& 2						& 0.139 					& Pol 				& 104976			& 	4	& 6 / 8	& 1.048 \\ \hline
		MRLS(36, 104976)	& $u=27$ 	& 3						& 0.194						& Pol 				& 104976			& 	4	& 6 / 8	& 1.561 \\ \hline
		MRLS(32, 16640)		& $u=19$ 	& 1.462					& 0.122						& Pol 				& 16640				& 4 				& 4 / 6 	& 0.900\\ \hline
		\hline
		OFT(36, 11052)	& $q=17$		& 1 					& 0.083						& Pol 				& 11052 			& 2 				& 4 	& 1	\\ \hline
		FT(36, 11664)	& $h=2$			& 2 					& 0.139						& MIN 				& 11664 			& 4					& 4		& 1	\\ \hline
		FT(36,  104976)	& 50\% pop.		& 3						& 0.222						& MIN 				& 104976			& 6					& 6	& 1\\ \hline
		\hline
		DF+(32, 16640)		& 65 groups & 1.5					& 0.127						& FPAR 				& 16640				& 3 				& 6 	& 1	\\ \hline
		DF(32, 16512) 	& 129 groups	& 1.5					& 0.125						& UGAL				& 16512				& 3 				& 6 & 1		\\ \hline
	\end{tabular}
\end{table*}
	\section{Experimental Setup}\label{sec:setup}

	To compare the performance of the different topologies we employ CAMINOS~\cite{CAMINOS}, an event-driven network simulator that operates at the flit level and accurately models the microarchitecture of network switches.
	The switch operates with a $2\times$ speedup and a random allocator, and each packet consists of 16 flits.
	It provides buffer space for 8 packets per virtual channel (VC) at the input ports and 4 packets per VC at the output ports.
	The MRLS networks are configured with the minimum number of VCs required to guarantee deadlock freedom—specifically, 3 or 4 VCs for maximum hop bounds of 6 or 8, respectively.
	In each comparison, the baseline topologies are initially provisioned with an equivalent number of VCs.
	However, because MRLS routing rarely reaches its maximum hop bound in practice, the competing topologies are generally allocated one fewer VC during the majority of the simulations.
	The exception to this configuration are the Dragonfly topologies, which are consistently equipped with 4 VCs~\footnote{Resources to fully reproduce our experiments, including the experimental setup and the MRLS topology files, are available at \url{https://github.com/alexcano98/MRLS-topology-reproducibility.git}}.

	\input{simulaciones-group/figura_11k}

	\subsection{Evaluated Networks}
	We evaluate the MRLS topology against two distinct families of network architectures considering three different number of servers:
		\begin{itemize}
			\item \textbf{Indirect Networks (FT and OFT):} Using a switch radix of 36, we evaluated an 11K-endpoint scale (both FT and OFT) and a 100K-endpoint scale (using a 50\% depopulated 4-level FT).
    		\item \textbf{Dragonfly Networks (DF and DF+~\cite{dragonfly_plus}):} Using a switch radix of 32, we evaluated a 16K-endpoint scale (full-sized DF, and DF+ with a global trunking of 4).
		\end{itemize}

	At each scale, one or more MRLS configurations are adapted to match the number of endpoints of the compared topologies, while also ensuring a similar or lower cost in terms of links per endpoint.
	State-of-the-art adaptive routing algorithms were utilized for all topologies.
	Polarized routing performs exceptionally well on MRLS; however, it is less effective (or requires extensive modifications) for DF and DF+.
	Thus, we employ the standard UGAL and FPAR mechanisms for the Dragonfly networks.
	Additionally, depending on the technology used, DF networks may achieve lower real-world costs by utilizing cheaper electrical cables for short links, a factor ignored in this simplified cost model defined in Equation~(\ref{eq:cost}).
	The details of the simulated topologies are presented in Table~\ref{tabla:topologias}.

	\subsection{Traffic scenarios}

	Evaluations were conducted under three distinct synthetic traffic models:

		\begin{enumerate}
		\item \textbf{Throughput (Max Injection Rate):} Uses Bernoulli equal-sized messages across four traffic patterns: Uniform (UN), Random Endpoint Permutation (REP), Random Switch Permutation (RSP---a highly adverse pattern), and Bipartite Uniform (BU), a pattern which models typical communication in a DC divided into two halves: one containing caches or directories, and the other accessing uniform data within them.

		\item \textbf{Tail Latency (Mice and Elephant Flows):} Measures up to the 99.99th percentile packet latency of uniform traffic at a load of 0.5 flits/cycle. Other load values have been evaluated as well, showing similar trends, but we omit them for brevity. To model real-world environments, traffic is split into mice flows (1-packet messages; 90\% of total messages, 10\% of volume) and elephant flows (16-packet messages; 10\% of total messages, 90\% of volume).

		\item \textbf{AI/HPC Collectives:} Evaluates All2All and Allreduce operations:
		\begin{itemize}
			\item \textbf{All2All:} Tasks are scaled to the minimum endpoints available per scale (11,052; 16,512; and 104,976 tasks).
   			\item \textbf{Allreduce:} Uses Rabenseifner's algorithm~\cite{thakur2005optimization}, which requires power-of-two task counts (8,192; 16,384; and 65,536 tasks). This scenario inherently favors structured topologies like FT, whose hypercube mapping facilitates the some locality characteristic of the algorithm. Additionally, message sizes vary between steps, and a substantial portion of the traffic is exchanged between endpoints connected to the same switch.
		\end{itemize}
		\end{enumerate}

	\section{Experimental Results}\label{sec:results}

	\input{simulaciones-group/figura_100k}

	\input{simulaciones-group/figura_directa}

	\subsection{Networks of 11K endpoints}
	In this subsection, we evaluate one OFT, one FT, and three MRLS networks from Table~\ref{tabla:topologias}, all of them comprising approximately 11K endpoints.
	Additionally, an MRLS configuration using KSP routing is evaluated to provide a comparison against Polarized routing.
	Results for the 11K-endpoint networks are presented in Figure~\ref{fig:routing_eval}.

	\subsubsection{Throughput evaluation}
	It can be observed that for the same cost, the MRLS topologies yield very similar performance to the OFT.
	While KSP routing in MRLS achieves similar performance to Polarized routing under most traffic patterns, its throughput drops by 26\% under the more challenging RSP traffic.
	Consequently, the remainder of our evaluation focuses on Polarized routing for MRLS due to its higher efficiency and robustness.

	For the more expensive configurations (MRLS with costs of 1.4 and 2.0, alongside the FT), throughput for UN and BU traffic is nearly identical, though MRLS holds a marginal advantage.
	Under REP traffic, the cost-2 MRLS outperforms the FT by 11\%, and the cost-1.4 MRLS by 24\%.
	Similarly, for the more challenging RSP traffic, the cost-2 MRLS outperforms the FT by 6\%, and the cost-1.4 MRLS by 40\%.

	The FT is a full-bisection bandwidth network with $\Theta=1$, which theoretically allows it to deliver maximum throughput across all traffic patterns.
	While the analogous MRLS lacks a formal proof of being rearrangeably non-blocking, its higher capacity ($\Theta$) reduces blocking across the fabric in practice, yielding an average speedup of 7.1\%.
	Notably, doubling the cost does not yield proportional throughput gains for UN and BU patterns, indicating that the network becomes over-provisioned for these workloads.

	\subsubsection{Mice and elephant flows evaluation}
	Regarding tail latency, the OFT and its equivalent MRLS perform nearly identically.
	However, the gap between the FT and its comparable MRLS widens considerably, with the FT exhibiting a roughly 30\% increase in packet latency.

	\subsubsection{Collectives evaluation}
	For both collectives, the comparable MRLS and OFT perform similarly, exhibiting only a 2\%--3\% difference in completion time.
	Conversely, the cost-2 MRLS outperforms the FT by 17\% in the All2All operation, whereas the FT outperforms MRLS by 10\% in the Allreduce.
	The FT's superiority in Allreduce stems from the strong locality of this communication pattern, which inherently benefits from the FT's hierarchical structure and its higher number of endpoints per switch.
	Furthermore, the cost-1.4 MRLS matches the All2All performance of the FT while reducing the cost by 30\%, although it remains 10\% slower in the Allreduce operation.

	\subsection{Networks of 100K endpoints}
	In this subsection, we compare one FT against three MRLS configurations from Table~\ref{tabla:topologias}, all scaled to 100K endpoints.
	The evaluation is presented in Figure~\ref{fig:throughput100k}.

	\subsubsection{Throughput evaluation}
	It can be observed that for the same cost, MRLS yields 4\%--7\% higher throughput than the FT.
	In this case, $\Theta(\text{MRLS})=1.561 > \Theta(\text{FT})=1$, suggesting that with ideal forwarding, a maximum throughput of 1.0 could be achieved.
	In practice, sporadic contention occurs, but it is less frequent in MRLS due to its higher capacity.

	Nevertheless, MRLS is over-dimensioned at $f=3$; employing $f=2$ yields comparable performance for most traffic.
	Specifically, this 33\% reduction in cost incurs only a 4\% penalty in Uniform throughput, though the performance penalty reaches 27\% under the more demanding switch permutation traffic.
	Below a cost of 2, MRLS throughput scales proportionally with cost across all traffic patterns (cf. \cite{Kassing}).
	Indeed, reducing the cost by 50\% (from $f=2$ to $f=1$) results in a throughput ratio between 0.48 and 0.52.

		\subsubsection{Mice and elephant flows evaluation}
	Tail latency metrics indicate that the MRLS with $f=1$ performs considerably worse than the other configurations, as it becomes saturated under the evaluated load.
	The remaining topologies exhibit very similar latencies, with the $f=3$ MRLS showing a slight advantage.

		\subsubsection{Collectives evaluation}
	The collective performance trends at 100K endpoints closely mirror those observed at 11K.
	MRLS is significantly faster for the All2All operation (by 48\%), whereas the FT excels in the Allreduce operation (by 20\%).
	This highlights that the extra capacity of MRLS does not strictly guarantee faster execution times, especially for highly structured traffic patterns that benefits other networks.
	Nevertheless, the proportionality of throughput to cost holds true in the All2All, providing superior performance per cost when the thickness factor $f$ is adequately chosen.

	\subsection{Comparison with Dragonfly networks}
	In this subsection, the performance of MRLS is compared against DF and DF+ topologies.
	The results for these networks are displayed in Figure~\ref{fig:throughput-directas}.
	The DF networks feature a diameter of 3 and a small average distance, reflecting a cost of $\costlinks \approx 1.5$.
	They are compared against an MRLS configuration with an equivalent cost ($f=\costlinks=1.46$).

	In principle, this is a disadvantageous scenario for MRLS, as its routing paths cannot have an odd integer length.
	When adjusting the MRLS to Dragonfly sizes, this limitation is reflected in its relatively low $\Theta=0.900$.

	\subsubsection{Throughput evaluation}
	The MRLS outperforms both Dragonfly variants across all traffic patterns while maintaining a similar links-per-endpoint ratio.
	For Uniform traffic, it yields 14\% and 19\% higher throughput than DF+ and DF, respectively.
	These performance gaps widen further under the other traffic patterns.

	The Bipartite Uniform (BU) traffic pattern is particularly noteworthy.
	While all indirect networks maintained similar performance for both UN and BU traffic, Dragonfly networks struggle here.
	In a BU scenario, where half the racks act as clients and the other half as caches, relying on minimal routes severely limits the number of usable global links in a Dragonfly topology.
	To exploit the full global capacity, these networks must use longer, non-minimal routes, a requirement that does not arise under Uniform traffic.
	This global capacity bottleneck affects all intra-group traffic.
	As a result, even a balanced network with $\Theta=1$ suffers severe performance degradation when traffic is not perfectly distributed across global links.
	The same effect occurs in RSP and REP, the latter to a lesser extent.
	This demonstrates that Dragonfly networks are highly sensitive to the specifics of process allocation; therefore, an optimal application placement would improve their performance on this benchmark.%

	\subsubsection{Mice and elephant flows evaluation}
	The tail latency evaluation shows that MRLS outperforms standard DF, while DF+ exhibits slightly better tail latency than MRLS.

	\subsubsection{Collectives evaluation}
	For the collective operations, MRLS demonstrates vastly superior performance.
	In the All2All operation, MRLS is approximately 100\% faster than both Dragonfly variants.
	For the Allreduce operation, MRLS is 50\% faster than DF and 12\% faster than DF+.
	While the direct nature of Dragonfly networks assists them in Allreduce, MRLS remains overwhelmingly superior for global All2All communication.
	This is once again caused by the inherent global bottleneck of Dragonfly topologies.

	\section{Conclusions}\label{sec:conclusions}

	Due to its high cost, alternatives to the Fat-Tree must be considered to achieve extreme scalability.
	Networks such as OFT have demonstrated their ability to interconnect a huge number of endpoints.
	However, they are tied to very specific sizes and their capacity for expansion is extremely limited.
	In contrast, MRLS networks are expandable and can accommodate virtually any number of endpoints.
	They enable a trade-off between scalability and performance, interconnecting a massive number of nodes without significant degradation.

	Furthermore, MRLS networks (and OFT) require routing algorithms that support non-minimal paths to handle adverse traffic efficiently.
	Challenges like its physical layout and implementation remain and merit further research for adoption in large-scale DCs.
	The MRLS achieves a 50\% speedup against a Fat-Tree for an All2All collective comprising 100k endpoints, and 100\% against Dragonfly networks for the same collective.

	\clearpage%
	\bibliography{main}
	\bibliographystyle{IEEEtran}

\appendix

\section{Average Distance and Probabilities}\label{apx:probability}
This appendix explains the mathematical details behind the graphs shown in Figures~\ref{fig:scal-spectrum} and \ref{fig:scalability}. Most of Figure~\ref{fig:scal-spectrum} can be estimated experimentally by generating many networks near the thresholds and measuring the probability of obtaining a given $D^*$. Indeed, we have also done that for the first three thresholds. Deriving analytical expressions that accurately match these measurements makes it possible to extrapolate to much larger network sizes. The same expressions also enable the construction of Figure~\ref{fig:scalability}, which would be otherwise unfeasible due to the vast amount of topologies to generate.

The main issue at hand is an accurate estimation of the distance distribution. Afterwards, the average distances $A$ and $A^*$ can be computed. Furthermore, additional details are provided to obtain the probability that separates the regions corresponding to different values of $D^*$. Our focus is on accurate predictions for realistic network sizes, where asymptotic approximations such as those in \cite{rfc} perform poorly. In Figure~\ref{fig:scal-spectrum}, such approximations could lead to important deviations. In contrast, the approximations presented here match the experimental data precisely.

For statistical notation, we denote the probability of an event $E$ by $P[E]$, the conditional probability given condition $C$ by $P[E|C]$, the expected value of a random variable $V$ by $E[V]$, and the conditional expectation of the random variable $V$ given a predicate $C$ by $E[V|C]$.

\subsection{Distance Distribution}\label{sub:distance-distribution}

Denote by $\mathcal S_r(c)$ the sphere of switches at distance exactly $r$ from a switch $c$, the center. The cardinalities of these spheres form the distance distribution. As we are dealing with random topologies, the sequence varies depending on the chosen center. Nevertheless, from a probabilistic point of view there are only two sequences ($i=1,2$):
$n^i_r=E[|\mathcal S_r(c) |\mid \text{$c$ is a switch at level $i$}]$.
The sequence $n^1_r$ is enough to compute $A$, both are required to determine the $D^*$ regions.
We employ $i=1,2$ for coherence with Table~\ref{tbl:notation}. At some point we employ arithmetic modulo 2, which must be understood as $i+r=1+\rem(1+i+r,2)=j$, the value $j=1,2$ with $i+r\equiv j \pmod 2$.

We also employ, for a switch $c$ at level $i$, the ball $\mathcal B_r(c)$ containing those switches at the level $r+i$ that are at distance $r$ or lower from $c$. Define the analogous sequences $b_r^i=E[|\mathcal B_r(c)|\text{$c$ is a switch at level $i$}]$. Note that $n_r^i\leq b_r^i\leq N_{r+i}$ and, for $r\geq 2$, $n_r^i=b_r^i-b_{r-2}^i$.

The first values are obtained directly from the regularity of MRLS, $n^1_0=n^2_0=b^1_0=b^2_0=1$, $n^1_1=b^1_1=u$, and $n^2_1=b^2_2=R$. Afterwards, we deal with $\mathcal B_r(c)$ as if it were the neighborhood of an arbitrary set. This is, for $1\leq r$, $\mathcal B_r(c)=\mathcal N(\mathcal B_{r-1}(c))$, where the neighborhood $\mathcal N$ is defined as $\mathcal N(X)=\bigcup_{a\in X}\mathcal S_1(a)$. Therefore,
\begin{equation}\label{eq:distance-distribution-iteration}
	b^i_{r+1}=\eta^{i+r}(b^i_r),
\end{equation}where $\eta^i(x)=E[|\mathcal N(X)| \mid |X|=x, \ \text{all switches of $X$ are at level $i$}]$.

Estimating the value of $|\mathcal N(X)|$ is a variant of the coupon collector problem. In particular, as a consequence of Theorem 2.1 in  \cite{Kan} it can be shown that

\begin{equation}\label{eq:eta}
	\eta^i(x)=N_{i+1}\left(1-\exp\left(-\frac{x n^i_1}{N_{i+1}}\right)\right).
\end{equation}
Then, iterating (\ref{eq:distance-distribution-iteration}) provides the expected count of switches at every distance. For the average distance, take only even distances starting at a leaf switch, this is,
\begin{equation*}
	A = \frac{1}{N_1-1}\sum_{i=1}^{\lfloor D/2\rfloor} 2i n^1_{2i}.
\end{equation*}

In the same way, $A^*$ can also be computed. In this case, without skipping any term and starting from both leaf and spine switches. Therefore, $A^* = (N-1)^{-1}\sum_{k=1}^{D^*} kn^1_k + kn^2_k$.

\subsection{Threshold Function for $D^*$}

A key figure of random networks is the value at which the diameter increases. In \cite{Bollobas} functions that separate a property such as the diameter are denoted as \textsl{threshold functions}.
That is an asymptotic definition and thus lies outside the scope of our focus.

Let us consider a pair of switches $s$ and $t$, and their distance $\delta=d(s,t)$. Then $t\in \mathcal S_\delta(s)$ and $s\in \mathcal S_\delta(t)$. Furthermore, for any $0\leq i,j\leq \delta$ with $i+j=\delta$, $\mathcal S_i(s)\cap \mathcal S_j(t)$ contains those switches that belong to any minimal path and are at the indicated step. Hence, to consider the diameter threshold, our aim is to determine whether $\mathcal S_i(s)\cap \mathcal S_j(t)$ is empty or not.

To obtain the threshold separating $D^*=k$ from $D^*=k+1$ we consider a pair $i,j$ with $0\leq i,j\leq k-1$ and $i+j=k-1$. If some candidates $s,t$ fail to be within distance $k-1$, then $d(s,t)\geq k+1$. Note that it cannot be $k$ due to parity. In other words, it is the threshold separating $k-1$ from $k+1$ for the maximum distance between switches at the corresponding levels. Asymptotically, it is unimportant which pair $i,j$ to choose. However, such a decision can lead to small but appreciable errors. The most precise choice appears to be $i=1$ and $s$ a leaf switch. Note how this forces the intersection to occur at the spine level.

For the calculation of the actual threshold we need to consider all pairs. Hence, let $G$ be the total count of possible pairs.
That is, let $G=\binom{N_1}{2}$ when both $s$ and $t$ are leaf (even distance) and $G=N_1N_2$ otherwise. Let $\lambda$ be the expected number of pairs $(s,t)$ with $\mathcal S_1(s)\cap \mathcal S_{k-2}(t)=\emptyset$.
Writing $\lambda$ in terms of the probability of intersection of the considered spheres,
\begin{multline}\label{eq:empty-pairs}
	\lambda = GP[\mathcal S_1(s)\cap \mathcal S_{k-2}(t)=\emptyset \mid\\
		\qquad\text{$s$ leaf, $t$ leaf if $k$ odd, spine if $k$ even}].
\end{multline}
This gives the probability of having a diameter below a given $k$,
\begin{equation}\label{eq:prob-all-empty}
	P[D^*\leq k] = (1-\frac{\lambda}{G})^G = \exp\left(-\lambda+O\left(\frac{1}{G}\right)\right).
\end{equation}

Some remarks are needed regarding the probability that $\mathcal S_1(s)$ and $\mathcal S_{k-1}(t)$ have an empty intersection, which is considered next.

\subsection{Probability for Empty Intersection}

To calculate the probability in $(\ref{eq:empty-pairs})$ we need to consider the size for $X=\mathcal S_1(s)$ and $Y=\mathcal S_{k-2}(t)$.
This is the distance distribution estimated in previous subsection~\ref{sub:distance-distribution}.

Consider the probability that two random sets $X, Y\subset \Omega$ have empty intersection.
For convenience, we take $x=|X|$, $y=|Y|$ and $n=|\Omega|$.
Let us consider $\Omega$ and $Y$ to be fixed, and count the ways that $X$ could be made.
From the $n$ elements of $\Omega$ choose $x$ to form $X$.
The set $X\cap Y$ is empty when these $x$ choices are taken over the $n-y$ elements of $\Omega\setminus Y$, that is:

\begin{equation}\label{eq:prob-intersection}
	P[X\cap Y=\emptyset] = \frac{\binom{n-y}{x}}{\binom{n}{x}} = \frac{(n-x)!(n-y)!}{(n-x-y)!n!}
\end{equation}
This expression is symmetric in $x$ and $y$, and therefore in $X$ and $Y$, as it can be observed in
the right side of (\ref{eq:prob-intersection}).

Using equations (\ref{eq:distance-distribution-iteration}),(\ref{eq:eta}),(\ref{eq:empty-pairs}),(\ref{eq:prob-all-empty}),(\ref{eq:prob-intersection}) yields values that fit experimental data, such as those drawn in Figure~\ref{fig:scal-spectrum}.

\end{document}